\documentclass[11pt]{article}
\usepackage{graphics}
\usepackage{latexsym}
\usepackage{algorithm}
\usepackage{algorithmic}
\usepackage{tikz}
\usepackage{soul}
\usepackage{mathtools}

\newcommand{\qed}{\hfill$\Box$}
\newcommand{\cO}{{\cal O}}
\newcommand{\cC}{{\cal C}}
\newcommand{\cA}{{\tt Ring}}

\newenvironment{proof}{\noindent {\bf Proof.}}{\qed}

\newtheorem{theorem}{Theorem}[section]
\newtheorem{lemma}{Lemma}[section]
\newtheorem{corollary}{Corollary}[section]

\newtheorem{proposition}{Proposition}[section]

\begin{document} 

\baselineskip 0.2in
\parskip      0.1in
\parindent    0em

\bibliographystyle{plain}

\title{{\bf Exploration of Faulty Hamiltonian Graphs}}

\author{
David Caissy \footnotemark[1]
\and
Andrzej Pelc \footnotemark[1] \footnotemark[2] 
}

\date{ }
\maketitle
\def\thefootnote{\fnsymbol{footnote}}

\footnotetext[1]{
\noindent
 D\'{e}partement d'informatique, Universit\'{e} du Qu\'{e}bec en Outaouais,
Gatineau, Qu\'{e}bec J8X 3X7,
 Canada. E-mails:
{\tt david.caissy@gmail.com}, {\tt pelc@uqo.ca}}

\footnotetext[2]{
\noindent
Research supported in part by the NSERC Discovery
Grant 8136 -- 2013 and by the
Research Chair in Distributed Computing of the
Universit\'{e} du Qu\'{e}bec en Outaouais. 
}

\begin{abstract}
We consider the problem of exploration of networks, some of whose edges are faulty. A mobile agent,
situated at a starting node and unaware of which edges are faulty, 
has to explore the connected fault-free component of this node
by visiting all of its nodes. The {\em cost} of the exploration is the number of edge traversals. 
For a given network and
given starting node, the {\em overhead} of an exploration algorithm is the worst-case ratio
(taken over all fault configurations)
of its cost to the cost of an optimal algorithm which knows where faults are situated.
An exploration algorithm, for a given network and given starting node, 
is called {\em perfectly competitive} if its overhead is the smallest among
all exploration algorithms not knowing the location of faults.
We design a perfectly competitive exploration algorithm for any ring, 
and show that, for networks modeled by hamiltonian graphs, 
the overhead of any DFS exploration is at most 10/9 times larger than that
of a perfectly competitive algorithm. Moreover, for hamiltonian graphs of size at least 24,
this overhead is less than 6\% larger than that of a perfectly competitive algorithm.

\vspace*{1cm}


\vspace*{5cm}
\end{abstract}

\pagebreak

\section{Introduction}

Exploration of networks by visiting all of its nodes is one of the basic tasks performed by mobile agents in networks.
In applications, a software agent may need to collect data residing at nodes of a network, or a mobile robot 
may need to collect samples of ground in a contaminated mine whose corridors form links of a network, with 
corridor crossings being nodes. In many situations, some links of the network may be faulty, preventing the mobile agent
to traverse them. In a computer network a fault may be a malfunctioning link, and in a cave it may be an obstructed corridor.
In this situation the agent may be unable to reach some nodes of the network, and the task is then to explore the accessible part,
as efficiently as possible.

The network is modeled as a simple connected
undirected graph $G=(V,E)$, called {\em graph} in the sequel.
The agent is initially situated at a starting node $v$ of the graph. It has a faithful labeled map
of the graph with its starting position marked. Each node has a distinct label, and ports at each node of degree $d$ 
are labeled by integers $1,\dots , d$.
Thus the agent knows which port at any node leads to which neighbor. 
 However, some of the edges of the graph are faulty.
The set of faulty edges $F \subseteq E$ is called a {\em fault configuration}.  The agent
does not know {\em a priori} the location of faults nor their number. When visiting a node for the
first time, the agent discovers ports corresponding to faulty edges, if any, incident to this node. A faulty edge prevents the agent
from traversing it. The fault-free subgraph $G'$, resulting from the original graph $G=(V,E)$ after removing all faulty edges,
may be disconnected. Let $C$ be the connected component of the subgraph $G'$,  containing the starting node $v$.
We will call $C$ the fault-free component corresponding to the fault configuration $F$ and to the node $v$.
The task of the agent is to {\em explore} $C$ by visiting all of its nodes.
Exploration is finished when the last node of $C$ is visited (the agent does not have to go back to its starting node).
For a given graph $G$, starting node $v$, and fault configuration $F$, 
the {\em cost} ${\cal C}(A,G,v,F)$ of an exploration algorithm $A$ is the number of edge traversals it
performs when exploring the connected component $C$ containing $v$ and corresponding to $F$.

An agent that knows $F$, and hence knows the component $C$, has some algorithm
exploring the component $C$ at the minimum possible cost. Call this optimal cost $opt(G,v,F)$.
Note that this optimal algorithm and its cost might be difficult to find in many cases. 
Now consider an exploration algorithm $A$ for a given graph $G$ and starting 
node $v$, that {\em does not know} $F$, as supposed in our scenario.
A natural measure of performance
of such an algorithm, cf. \cite{MaPe}, is the worst-case ratio between its cost and the optimal
cost $opt(G,v,F)$, where the worst case is taken over all fault configurations $F$. This number
$\mbox{max}_{F \subseteq E} \frac{{\cal C}(A,G,v,F)}{opt(G,v,F)}$ is called the {\em overhead} of $A$,
and is denoted ${\cal O}_{A,G,v}$. This measure is similar to the competitive ratio of 
on-line algorithms. It measures the penalty incurred by the algorithm, due to some kind of ignorance.
In the case of on-line algorithms, they do not know future events, 
which are known to an off-line algorithm (serving as a benchmark). 
The exploration algorithms we want to design do not know the fault configuration, which is
known to an optimal algorithm (serving as a benchmark).

Given a graph $G$ and a starting node $v$,
an exploration algorithm (not knowing the fault configuration) 
is called {\em perfectly competitive} if it has the smallest overhead among all exploration algorithms
working under this scenario. Our aim is to construct exploration algorithms with small overhead:
either perfectly competitive algorithms or ones whose overhead only slightly exceeds the smallest
possible overhead.

In the sequel, when the graph $G$ and the starting node $v$ are fixed, we will write
${\cal C}(A,F)$ instead of ${\cal C}(A,G,v,F)$, $opt(F)$ instead of $opt(G,v,F)$, and
${\cal O}_{A}$ instead of ${\cal O}_{A,G,v}$. If the link corresponding to port $p$ at some node is faulty,
we will say that port $p$ at this node is {\em faulty}, otherwise we will say that port $p$ is {\em free}.

\subsection{Related work}

The problem of 
exploration and navigation of mobile agents in an unknown environment
has been extensively studied in the literature (cf. the survey \cite{RKSI}).
Our scenario falls in this category:  in our case, the unknown
ingredient of the environment is the fault configuration.
The explored environment has been modeled in the literature in two distinct ways: 
either as a geometric terrain in the plane, e.g., an 
unknown terrain with convex obstacles \cite{BRS},
or a room with polygonal \cite{DKP} or rectangular \cite{BBFY} obstacles, or 
as we do, i.e., 
as a graph, assuming that the agent may only move along its edges. The graph
model can be further specified in two different ways:
either the graph is directed and strongly connected, in which case the agent can move only from
tail to head of a directed edge  \cite{AH,BFRSV,BS,DP}, or the graph is undirected (as we assume)
and the agent can traverse edges in both directions  \cite{ABRS,BRS2,DKK,PaPe}.
The efficiency measure adopted in most papers dealing with exploration of graphs is the cost
of completing this task, measured by the number of edge traversals by the agent.
In some papers, further restrictions on the moves of the agent are imposed.
It is assumed that the robot has either a restricted tank \cite{ABRS,BRS2},
and thus has to periodically return to the base for refueling, or that it is attached to the
base by a rope or cable of restricted length \cite{DKK}. 

Another direction of research concerns
exploration of anonymous graphs.
In this case it is impossible
to explore arbitrary graphs and stop after exploration, if no marking of nodes is allowed. 
Hence some authors \cite{BFRSV,BS}
allow {\em pebbles} which the agent can drop on nodes to recognize already visited ones, and
then remove them and drop them in other places. A more restrictive scenario assumes
a stationary token that is fixed at the starting node of the agent \cite{CDK,PeTi}. 
Exploring
anonymous graphs without the possibility of marking nodes (and thus possibly without stopping) 
is investigated, e.g., in \cite{DFKP}. 
The authors 
concentrate attention not on the cost of exploration but on the minimum amount of memory sufficient
to carry out this task.
Exploration of anonymous graphs is also considered in \cite{DM,DJMW,DW}.

A measure of performance of exploration algorithms, similar to the competitive ratio, and thus to our
notion of overhead, has been used, e.g., in \cite{DePe}, to study 
exploration of unknown graphs and graphs
for which only an unoriented map is available. In the above paper, the benchmark was the performance
of an algorithm having full knowledge of the graph. 

Our problem belongs to the domain of fault-tolerant graph exploration whose other aspects were studied, e.g., in
\cite{CKR1,CKR2,CKMP,DFKRPS,DFPS1,DFPS2,FKMS,KMRS1,KMRS2}, in the context of searching for a black hole in a network. A black hole is a process situated at
an unknown node and destroying all agents visiting it. The goal is to locate a black hole using as few agents as possible, while keeping
at least one agent surviving. Another aspect of fault-tolerant exploration was studied in \cite{FPS}. It concerned the exploration of a line by a single agent, and faults
involved directions chosen by the agent: for example, if the agent planned on going right and the fault occurred, it resulted in going left.

The paper most closely related to the present work is \cite{MaPe}. The authors used the same scenario and the goal was also to find
exploration algorithms with the smallest possible overhead. However, they restricted attention to the class of trees. For the simplest case of the line
they designed a perfectly competitive algorithm, and they proposed another algorithm,  working for arbitrary trees, whose overhead is 
at most 9/8 times larger than that of a perfectly competitive algorithm.

\subsection{Our results}

We consider hamiltonian graphs, i.e., graphs that contain a simple cycle including all nodes.
This is a large class containing such important examples as complete graphs, hypercubes, rings, and even-size tori.
Our goal is to design exploration algorithms for hamiltonian graphs some of whose edges are faulty, with small overhead. We have two main results.
\begin{itemize}
\item
For any ring, and any starting node, 
we design a perfectly competitive exploration algorithm. The algorithm  and its
overhead depend only on the size of the ring. This should be contrasted with the perfectly competitive exploration algorithm
for the line, designed in \cite{MaPe}, whose behavior and overhead depend on the distance between the starting node and the closer endpoint of the line,
and {\em do not} depend on the size of the line.
\item
For an arbitrary hamiltonian graph of size $n\geq 3$ and an
arbitrary starting node, we show that the overhead of any DFS exploration is at most $(2n-4)/(n-1)$. We also show that this overhead is at most 10/9 times larger than that
of a perfectly competitive algorithm, for any hamiltonian graph. Moreover, for hamiltonian graphs of size at least 24,
this overhead turns out to be less than 6\% larger than that of a perfectly competitive algorithm.
\end{itemize}
The total computation time used by our exploration algorithms is linear in the number of edges of the
explored graph. Our main contribution is in the analysis of overhead of these algorithms, showing that
simple and natural exploration strategies perform very well in faulty hamiltonian graphs. 


\section{Perfectly competitive exploration of rings}

In this section we design a perfectly competitive algorithm working for an arbitrary ring.
A ring is a graph
$R = (V,E)$, where $V = \{v_1, ...,v_n\}$, $n
\geq 3$ and $E = \{\{v_i, v_{i+1}\} : i=1, ..., n-1\} \cup \{ \{v_n, v_1\}\}$.
For any $i>1$ we call $v_{i-1}$ the {\em predecessor} of $v_i$ and for any $i<n$ we call
$v_{i+1}$ the {\em successor} of $v_i$. The predecessor of $v_1$ is $v_n$ and the successor of $v_n$ is $v_1$.
At any node we denote by $\ell$ the port leading to the predecessor of $v$ and by $r$ the port leading to the successor of $v$. 

We fix a ring $R$ and a starting node $v$. For a fixed non-empty fault configuration $F$, denote by $x$ the distance between $v$ and the first fault encountered when always taking the port $\ell$,
and denote by $y$ the distance between $v$ and the first fault encountered when always taking the port $r$. 

An agent knowing the configuration $F$, and hence knowing $x$ and $y$, has the following simple exploration algorithm which is optimal:
go first to the closest fault and then go back until a fault is met. Hence $opt(F)$ is $2x + y$ if $x \leq y$, and is $2y + x$ if
$x > y$.
	
The following procedure has two parameters: a positive integer $s$ and a {\em direction} $d \in \{\ell,r\}$. It causes the agent to
go at most $s$ steps always taking port $d$, until a fault is met. It returns the boolean value $b=false$ if a fault is met and $b=true$ otherwise.

\begin{center}
\fbox{
\begin{minipage}{15cm}

{\bf Procedure} GO-AT-DISTANCE ($s$, $d$)

\hspace*{1cm} {\bf while} port $d$ at the current node is free and travelled distance is $< s$ {\bf do}\\
\hspace*{1cm} \hspace*{1cm} take port $d$\\
\hspace*{1cm} {\bf if} port $d$ at the current node  is free {\bf then} $b:=true$ {\bf else} $b:= false$

\end{minipage}
}
\end{center}

The next procedure causes the agent to successively take port $d \in \{\ell,r\}$ until a fault is met.
It has a provision to avoid returning to the starting node $v$, if there are no faults.

\begin{center}
\fbox{
\begin{minipage}{13cm}

{\bf Procedure} GO-FIRM($d$)

\hspace*{1cm} {\bf while} port $d$ at the current node  is free and 
fewer than $n$ nodes\\
\hspace*{1cm} have been visited {\bf do}\\
\hspace*{1cm} \hspace*{1cm} take port $d$

\end{minipage}
}
\end{center}

We now formulate our algorithm for ring exploration.

\begin{center}
\fbox{
\begin{minipage}{11cm}

{\bf Algorithm} {\tt Ring}\\

\hspace*{1cm} {\bf if} port $\ell$ is faulty {\bf then} GO-FIRM($r$)\\
\hspace*{1cm} {\bf else}\\
\hspace*{1cm} \hspace*{1cm} {\bf if} port $r$ is faulty {\bf then} GO-FIRM($\ell$)\\
\hspace*{1cm} \hspace*{1cm} {\bf else}\\
\hspace*{1cm} \hspace*{1cm} \hspace*{1cm} {\bf if} $n \leq 5$ {\bf then}\\
\hspace*{1cm} \hspace*{1cm} \hspace*{1cm} \hspace*{1cm} GO-FIRM($\ell$)\\
\hspace*{1cm} \hspace*{1cm} \hspace*{1cm} \hspace*{1cm} GO-FIRM($r$)\\
\hspace*{1cm} \hspace*{1cm} \hspace*{1cm} {\bf else}\\
\hspace*{1cm} \hspace*{1cm} \hspace*{1cm} \hspace*{1cm} {\bf if} $6 \leq n \leq 19$ {\bf then}\\
\hspace*{1cm} \hspace*{1cm} \hspace*{1cm} \hspace*{1cm} \hspace*{1cm}  GO-AT-DISTANCE(1, $\ell$)\\
\hspace*{1cm} \hspace*{1cm} \hspace*{1cm} \hspace*{1cm} \hspace*{1cm} {\bf if} $b$ {\bf then}\\
\hspace*{1cm} \hspace*{1cm} \hspace*{1cm} \hspace*{1cm} \hspace*{1cm} \hspace*{1cm} GO-FIRM($r$)\\
\hspace*{1cm} \hspace*{1cm} \hspace*{1cm} \hspace*{1cm} \hspace*{1cm} \hspace*{1cm} GO-FIRM($\ell$)\\
\hspace*{1cm} \hspace*{1cm} \hspace*{1cm} \hspace*{1cm} \hspace*{1cm} {\bf else}\\
\hspace*{1cm} \hspace*{1cm} \hspace*{1cm} \hspace*{1cm} \hspace*{1cm} \hspace*{1cm} GO-FIRM($r$)\\
\hspace*{1cm} \hspace*{1cm} \hspace*{1cm} \hspace*{1cm} {\bf else}\\ 
\hspace*{1cm} \hspace*{1cm} \hspace*{1cm} \hspace*{1cm} \hspace*{1cm}  GO-AT-DISTANCE(2, $\ell$)\\
\hspace*{1cm} \hspace*{1cm} \hspace*{1cm} \hspace*{1cm} \hspace*{1cm} {\bf if} $b$ {\bf then}\\
\hspace*{1cm} \hspace*{1cm} \hspace*{1cm} \hspace*{1cm} \hspace*{1cm} \hspace*{1cm} GO-FIRM($r$)\\
\hspace*{1cm} \hspace*{1cm} \hspace*{1cm} \hspace*{1cm} \hspace*{1cm} \hspace*{1cm} GO-FIRM($\ell$)\\
\hspace*{1cm} \hspace*{1cm} \hspace*{1cm} \hspace*{1cm} \hspace*{1cm} {\bf else}\\
\hspace*{1cm} \hspace*{1cm} \hspace*{1cm} \hspace*{1cm} \hspace*{1cm} \hspace*{1cm} GO-FIRM($r$)\\

\end{minipage}
}
\end{center}

Since the time of computation at each visit of a node is constant, the total computation time of Algorithm {\tt Ring}
is linear in the size of the ring.
The following proposition gives the overhead of Algorithm {\tt Ring}.

\begin{proposition}\label{ring}
The overhead of Algorithm {\tt Ring} is:

\

\begin{tabular}{rll}
	$\bullet$ & $1$ 					& when\ \ $n = 3$ \\
	\\
	$\bullet$ & $\frac{2n - 3}{n}$ 	    & when\ \ $4 \leq n \leq 5$ \\
	\\
	$\bullet$ & $\frac{3}{2}$ 			& when\ \ $6 \leq n \leq 7$ \\
	\\
	$\bullet$ & $\frac{2n - 2}{n + 1}$  & when\ \ $8 \leq n \leq 19$ \\
	\\
	$\bullet$ & $\frac{9}{5}$ 			& when\ \ $20 \leq n \leq 23$ \\
	\\
	$\bullet$ & $\frac{2n - 1}{n + 2}$  & when\ \ $n \geq 24$ \\	
\end{tabular} 
\end{proposition}

\begin{proof}
First observe that if the fault configuration $F$ is empty, then Algorithm {\tt Ring} has cost $n-1$ if $n \leq 5$,  cost $n$ if $6 \leq n \leq 19$, and cost $n+1$ if $n \geq 20$. 
Hence, for the empty fault configuration $F$,
$\frac{\cC ( \cA, F )}{opt(F)}=1$  if $n \leq 5$, $\frac{\cC ( \cA, F )}{opt(F)}=\frac{n}{n-1}$  if $6 \leq n \leq 19$, and $\frac{\cC ( \cA, F )}{opt(F)}=\frac{n+1}{n-1}$ if $n \geq 20$. 
If $F$ is non-empty then integers $x$ and $y$ are well defined.  
If $x = 0$ or $y = 0$, then $\cal{C}$$(\cA, F) = opt(F)$. Otherwise consider the following cases.

$n = 3$. In this case $\cal{C}$$(\cA, F) = 2x + y = 2y + x = opt(F)$. Hence in this case $\cO _{\cA}=1$.

$4 \leq n \leq 5$. In this case, $\frac{{\cal{C}}(\cA, F)}{opt(F)} = \frac{2x +
y}{min\{2x + y, 2y + x\}}$.
If $x \leq y$, then $\frac{{\cal{C}}(\cA, F)}{opt(F)} = 1$. If $x > y$, then
$\frac{{\cal{C}}(\cA, F)}{opt(F)} = \frac{2x + y}{2y + x}$ and this fraction is maximized for
$x = n - 2$, $y = 1$, giving the value
 $\frac{{\cal{C}}(\cA, F)}{opt(F)} = \frac{2n - 3}{n}$. Since $1 < \frac{2n -
3}{n}$ for $n \geq 4$,  we get $\cO _{\cA} = \frac{2n - 3}{n}$ in this case.

$n = 6$. In this case, $\frac{{\cal{C}}(\cA, F)}{opt(F)} = \frac{2 + 2y +
x}{min\{2x + y, 2y + x\}}$.
If $x < y$, $\frac{{\cal{C}}(\cA, F)}{opt(F)} = \frac{2 + 2y + x}{2x + y}$ and this fraction is maximized for
$x = 2$, $y = 3$, giving the value
 $\frac{{\cal{C}}(\cA, F)}{opt(F)} = \frac{10}{7}$. If $x \geq y$, then
$\frac{{\cal{C}}(\cA, F)}{opt(F)} = \frac{2 + 2y + x}{2y + x}$ and this fraction is maximized for
$x = 2$, $y = 1$, giving the value
 $\frac{{\cal{C}}(\cA, F)}{opt(F)} = \frac{3}{2}$. The ratio $\frac{{\cal{C}}(\cA, F)}{opt(F)}$ for the empty fault configuration is $\frac{6}{5}$ in this case. Since $\frac{10}{7} <
\frac{3}{2}$ and $\frac{6}{5}<\frac{3}{2}$,  we get $\cO _{\cA} = \frac{3}{2}$ in this case.

$n = 7$. In this case, $\frac{{\cal{C}}(\cA, F)}{opt(F)} = \frac{2 + 2y +
x}{min\{2x + y, 2y + x\}}$.
If $x \leq y$, then $\frac{{\cal{C}}(\cA, F)}{opt(F)} = \frac{2 + 2y + x}{2x + y}$ and this fraction is maximized for
 $x = 2$, $y = 4$, giving the value
 $\frac{{\cal{C}}(\cA, F)}{opt(F)} = \frac{3}{2}$. If $x > y$, then
$\frac{{\cal{C}}(\cA, F)}{opt(F)} = \frac{2 + 2y + x}{2y + x}$ and this fraction is maximized for
$x = 2$, $y = 1$, giving the value
 $\frac{{\cal{C}}(A, F)}{opt(F)} = \frac{3}{2}$. The ratio $\frac{{\cal{C}}(\cA, F)}{opt(F)}$ for the empty fault configuration is $\frac{7}{6}$ in this case. Since $\frac{7}{6}<\frac{3}{2}$,  we get
$\cO _{\cA}= \frac{3}{2}$ in this case.

$8 \leq n \leq 19$. In this case, $\frac{{\cal{C}}(\cA, F)}{opt(F)} = \frac{2 + 2y
+ x}{min\{2x + y, 2y + x\}}$.
If $x \leq y$, then $\frac{{\cal{C}}(\cA, F)}{opt(F)} = \frac{2 + 2y + x}{2x + y}$ and this fraction is maximized for
$x = 2$, $y = n - 3$, giving the value
 $\frac{{\cal{C}}(\cA, F)}{opt(F)} = \frac{2n - 2}{n + 1}$. If $x > y$, then
$\frac{{\cal{C}}(\cA, F)}{opt(F)} = \frac{2 + 2y + x}{2y + x}$ and this fraction is maximized for
$x = 2$, $y = 1$, giving the value
 $\frac{{\cal{C}}(\cA, F)}{opt(F)} = \frac{3}{2}$. 
 The ratio $\frac{{\cal{C}}(\cA, F)}{opt(F)}$ for the empty fault configuration is $\frac{n}{n-1}$ in this case.
 Since $\frac{2n - 2}{n + 1}
> \frac{3}{2}$ and $\frac{2n - 2}{n + 1}
> \frac{n}{n-1}$  when $n \geq 8$, we get $\cO _{\cA} = \frac{2n - 2}{n + 1}$ in this case.

$20 \leq n \leq 23$. In this case, $\frac{{\cal{C}}(\cA, F)}{opt(F)} = \frac{4 + 2y
+ x}{min\{2x + y, 2y + x\}}$.
If $x \leq y$, then $\frac{{\cal{C}}(\cA, F)}{opt(F)} = \frac{4 + 2y + x}{2x + y}$ and this fraction is maximized for
$x = 3$, $y = n - 4$, giving the value
 $\frac{{\cal{C}}(\cA, F)}{opt(F)} = \frac{2n - 1}{n + 2}$. If $x > y$, then
$\frac{{\cal{C}}(\cA, F)}{opt(F)} = \frac{4 + 2y + x}{2y + x}$ and this fraction is maximized for
$x = 3$, $y = 1$, giving the value
 $\frac{{\cal{C}}(\cA, F)}{opt(F)} = \frac{9}{5}$. 
The ratio $\frac{{\cal{C}}(\cA, F)}{opt(F)}$ for the empty fault configuration is $\frac{n+1}{n-1}$ in this case. 
 Since $\frac{2n - 1}{n + 2}
\leq \frac{9}{5}$ and $\frac{9}{5}> \frac{n+1}{n-1}$ when $20 \leq n \leq 23$, we get $\cO _{\cA} = \frac{9}{5}$ in this case.

$n \geq 24$. In this case, $\frac{{\cal{C}}(\cA, F)}{opt(F)} = \frac{4 + 2y +
x}{min\{2x + y, 2y + x\}}$.
If $x \leq y$, then $\frac{{\cal{C}}(\cA, F)}{opt(F)} = \frac{4 + 2y + x}{2x + y}$ and this fraction is maximized for
$x = 3$, $y = n - 4$, giving the value
 $\frac{{\cal{C}}(\cA, F)}{opt(F)} = \frac{2n - 1}{n + 2}$. If  $x > y$, then
$\frac{{\cal{C}}(\cA, F)}{opt(F)} = \frac{4 + 2y + x}{2y + x}$ and this fraction is maximized for
$x = 3$, $y = 1$, giving the value
 $\frac{{\cal{C}}(\cA, F)}{opt(F)} = \frac{9}{5}$. 
The ratio $\frac{{\cal{C}}(\cA, F)}{opt(F)}$ for the empty fault configuration is $\frac{n+1}{n-1}$ in this case.  
 Since $\frac{2n - 1}{n + 2}
> \frac{9}{5}$ and  $\frac{2n - 1}{n + 2}>\frac{n+1}{n-1}$  when $n \geq 24$, we get $\cO _{\cA}= \frac{2n - 1}{n + 2}$ in this case.
\end{proof}

We now show that Algorithm {\tt Ring} is perfectly competitive, i.e.,  has the smallest overhead among all exploration
algorithms for the ring, not knowing the fault configuration. We adapt to the ring the following notions from
\cite{MaPe} (where they were defined for the line in a similar way). We denote by $\cal{A}$$_{k}$ the class of
exploration algorithms for the ring which do initially $k$ returns (i.e., changes of direction) assuming that no fault 
 is encountered before the first $k$ returns, then GO-FIRM, and in the case when a fault is met, return and GO-FIRM.
An algorithm
of class $\cal{A}$$_{k}$ is called an $i$-step algorithm, for 
$1\leq i \leq n-2$, if it goes $i$ steps before the first return, unless a fault is encountered earlier. Note that $\cal{A}$$_{0}$ is the class of DFS algorithms (there are two such algorithms for the ring, depending on the initial direction chosen). Also, an $i$-step algorithm of class $\cal{A}$$_{k}$, for $k>0$ and $i \geq n-1$ is in fact a
DFS algorithm, i.e., it is of class $\cal{A}$$_{0}$. On the other hand, a 0-step algorithm of class $\cal{A}$$ _k$ is an algorithm of a class $\cal{A}$$_m$, for some $m<k$, and hence it is enough
to consider only $i$-step algorithms, for strictly positive $i$, in each class.   
Algorithm {\tt Ring} is an algorithm of class 
$\cal{A}$$_{0}$ for
$n \leq 5$, it is a 1-step algorithm of class $\cal{A}$$_{1}$ for $6 \leq n
\leq 19$, and it is a 2-step algorithm of class $\cal{A}$$_{1}$ for $n \geq 20$.

The proof that Algorithm {\tt Ring} is perfectly competitive is split into a series of lemmas.

\begin{lemma}\label{class 0}
Every algorithm of class $\cal{A}$$_{0}$ has overhead at least 
$\frac{2n - 3}{n}$, for all $n \geq 3$.
\end{lemma}

\begin{proof}
Consider the fault configuration $F$ consisting of a single fault situated at distance 1 from the starting node $v$,
in the direction different from the one at which the algorithm $A$ of class $\cal{A}$$_{0}$ starts. Then
$\cal{C}$$(A, F) =2n-3$ and $opt(F)=n$, which concludes the proof.
\end{proof}

\begin{lemma}\label{class 1}
Let $n \geq 3$. Every $(n-2)$-step algorithm of class $\cal{A}$$_{1}$ has overhead $\frac{2n-3}{n-1}$, and  
every $i$-step algorithm of class $\cal{A}$$_{1}$, for $1
\leq i \leq n-3$, has overhead $max\{\frac{n+i-1}{n-1},\frac{3i
+ 3}{i + 3}, \frac{i + 2n - 3}{i + n}\}$.
\end{lemma}

\begin{proof}
First consider an $(n-2)$-step algorithm $A$ of class $\cal{A}$$_{1}$, and a fault configuration $F$.
If $F$ is empty, then ${\cal{C}}(A,F)=2n-3$ and $opt(F)=n-1$, giving the ratio $\frac{2n-3}{n-1}$.
Otherwise, $x$ and $y$ are well defined. If $x=0$ or $y=0$, then ${\cal{C}}(A, F)$ = $opt(F)$.
Assume that $x>0$ and $y>0$. Without loss of
generality we may assume that in the first step the algorithm takes port $\ell$. 
We have $x \leq n-2$, and hence $\frac{{\cal{C}}(A, F)}{opt(F)} = \frac{2x + y}{min\{2x + y,
2y + x\}}$. This fraction is maximized for $x = n-2$ and $y = 1$, and then its value is $\frac{2n-3}{n}$.
Since $\frac{2n-3}{n-1}>\frac{2n-3}{n}$, we get ${\cal O}_A= \frac{2n-3}{n-1}$.

In the rest of the proof we consider an $i$-step algorithm $A$ of class $\cal{A}$$_{1}$, for $1
\leq i \leq n-3$, and a fault configuration $F$.
If $F$ is empty, then ${\cal{C}}(A,F)=n+i-1$ and $opt(F)=n-1$, giving the ratio
of $\frac{n+i-1}{n-1}$.
Otherwise, $x$ and $y$ are well defined. If $x=0$ or $y=0$, then ${\cal{C}}(A, F)$ = $opt(F)$. Assume
that $x>0$ and $y>0$. Without loss of
generality we may assume that in the first step the algorithm takes port $\ell$.

If $x \leq i$ then $\frac{{\cal{C}}(A, F)}{opt(F)} = \frac{2x + y}{min\{2x + y,
2y + x\}}$, and this fraction is maximized for $x = i$ and $y = 1$. Hence the maximum value of
$\frac{{\cal{C}}(A, F)}{opt(F)}$ is $\frac{2i + 1}{i + 2}$ in this case.

If $x > i$ and $ x \geq y$ then $\frac{{\cal{C}}(A, F)}{opt(F)} = \frac{2i + 2y
+ x}{2y + x}$, and this fraction is maximized for $x = i + 1$ and $y = 1$.
 Hence the maximum value of
$\frac{{\cal{C}}(A, F)}{opt(F)}$ is $\frac{3i + 3}{i + 3}$ in this case.

If $x > i$ and $ x < y$ then 
$\frac{{\cal{C}}(A, F)}{opt(F)} = \frac{2i + 2y + x}{2x + y}$, and this fraction is maximized for $x = i + 1$ and
$y = n - i - 2$. Hence the maximum value of $\frac{{\cal{C}}(A, F)}{opt(F)}$ is 
$\frac{i + 2n - 3}{i + n}$ in this case.

Since $\frac{3i + 3}{i + 3} > \frac{2i + 1}{i + 2}$, il follows that
${\cal O}_A = \max\{\frac{n+i-1}{n-1},\frac{3i + 3}{i + 3}, \frac{i + 2n - 3}{i
+ n}\}$.
\end{proof}

Lemma \ref{class 1} gives the following corollary.

\begin{corollary}\label{cor} 
Every 1-step algorithm of class $\cal{A}$$_{1}$ has overhead:

\begin{tabular}{rll}
	$\bullet$ & $\frac{3}{2}$ 			& for\ \ $3 \leq n \leq 7$ \\
	\\
	$\bullet$ & $\frac{2n - 2}{n + 1}$  & for\ \ $n > 7$ \\
\end{tabular}

Every 2-step algorithm of class $\cal{A}$$_{1}$ has overhead: 

\begin{tabular}{rll}
	$\bullet$ & {$\frac{5}{3}$}  			& {for\ \ $n = 4$} \\
	\\
	$\bullet$ & $\frac{9}{5}$  			& for\ \ $5 \leq n \leq 23$ \\
	\\
	$\bullet$ & $\frac{2n - 1}{n + 2}$  & for\ \ $n > 23$ \\	
\end{tabular} 
\end{corollary}

Notice that for $3\leq i \leq n-3$, we have $\frac{3i + 3}{i + 3} \geq 2$.
Moreover,  $\frac{2n-3}{n-1}
\geq \frac{3}{2}$ when $5 \leq n \leq 7$, and $\frac{2n-3}{n-1} \geq \frac{2n -
2}{n + 1}$ when $n > 7$ ($i$ would be less than 3 for $n < 5$).
Hence it follows from Lemma \ref{class 1} and from Corollary \ref{cor} that the
overhead of any $i$-step algorithm of class $\cal{A}$$_{1}$, for $i \geq 3$, is
larger than the overhead of a $1$-step algorithm of class $\cal{A}$$_{1}$.
Hence, when looking for perfectly competitive algorithms in class
$\cal{A}$$_{1}$ we may restrict attention to $i$-step algorithms for $i \leq 2$.
 
It was observed in \cite{MaPe} that if the agent, at some point of the exploration, is at node 
$w$, then moves along an already traversed edge incident to $w$, and immediately returns 
to $w$, then, for any fault configuration, an algorithm causing such a pair of moves has cost 
strictly larger than the algorithm that skips these two moves. Hence, while looking for a perfectly competitive algorithm,
we may restrict attention to algorithms that do not perform such unnecessary moves. We call such algorithms {\em regular}.
It is easy to see that any regular exploration algorithm on a ring falls in one of the classes $\cal{A}$$_{k}$.   

The following lemma is proved similarly as Lemma 2.5 from \cite{MaPe}
which is its analog for the line;  we include the proof for completeness. Together with the above remarks, it shows that when looking 
for perfectly competitive exploration algorithms on rings we may restrict attention to algorithms
of class $\cal{A}$$_{0}$
and to $1$-step and $2$-step algorithms of class $\cal{A}$$_{1}$.

\begin{lemma}\label{old}
For every  exploration algorithm of class $\cal{A}$$_{k}$, for  $k \geq 2$, there exists an algorithm of class $\cal{A}$$_{0}$
or an $i$-step algorithm of class $\cal{A}$$_{1}$, for  $i \leq 2$, with a smaller or equal overhead.
\end{lemma}

\begin{proof}
Let $A$ be any algorithm of class $\cal{A}$$_{k}$, for $k \geq 2$. Without loss of generality, 
assume that the agent starts by taking port $\ell$. 
The behavior of $A$ can be
described as follows:

\qquad traverse $z_{1}$ edges going in one direction;

\qquad return and traverse $z_{1} + z_{2}$ edges;

\qquad return and traverse $z_{2} + z_{3}$ edges;

\qquad ...

\qquad return and traverse $z_{k - 2} + z_{k - 1}$ edges;

\qquad return and traverse $z_{k - 1} + z_{k}$ edges;

\qquad return and GO-FIRM;

\qquad return and GO-FIRM;

for $z_{1} < z_{3} < z_{5} < ...$ and $z_{2} < z_{4} < z_{6} < ...$, provided
that the closest faults to the starting node $v$ (if any) are at distance larger
than $z_{k - 1}$ of $v$ on the side where the agent has traversed $z_{k - 1}$
edges,  and at distance larger than $z_{k }$ of $v$ on the side where the agent
has traversed $z_{k }$ edges. This implies that $z_{k - 1} + z_{k} $
{ $\leq $} $ n - 2$.

First consider a non-empty fault configuration $F$ for which the closest fault
from $v$ in one direction is at distance $z_{k - 1} + 1$ from it, and the
closest fault from $v$ in the other direction is at distance $z_{k} + 1$ from it
(this could be the same fault). We have ${\cal{C}}(A, F) = 2(z_{1} + ... +
z_{k}) + (2z_{k - 1} + 2) + (z_{k} + 1) = 2(z_{1} + ... + z_{k}) + 2z_{k - 1} +
z_{k} + 3$.

If $z_{k} \geq z_{k - 1}$, we have $opt(F) =  2z_{k - 1} + z_{k} + 3$.
Hence $\frac{{\cal{C}}(A, F)}{opt(F)} =
\frac{2(z_{1} + ... + z_{k}) + 2z_{k - 1} + z_{k} + 3}{2z_{k - 1} + z_{k} + 3}$.
If additionally $2(z_{1} + ... + z_{k}) - 2z_{k - 1} - z_{k} \geq 3$, it follows that
$\frac{{\cal{C}}(A,P)}{opt(F)} \geq 2$ in this case.

If $z_{k} < z_{k - 1}$ then  $2(z_{1} + ... + z_{k - 1}) \geq z_{k} + 3$
and $opt(A)= 2z_{k} + z_{k - 1} + 3$. Hence
$\frac{{\cal{C}}(A, F)}{opt(F)} =
\frac{2(z_{1} + ... + z_{k}) + 2z_{k - 1} + z_{k} + 3}{2z_{k} + z_{k - 1} + 3}
\geq 2$ in this case.

Since any algorithm $A^{\prime}$ of class $\cal{A}$$_{0}$ has an overhead
$\cal{O}$$(A^{\prime}) < 2$, we conclude that in any of the above cases, the algorithm
$A^{\prime}$ has an overhead smaller than the algorithm $A$.

If none of the above cases holds (i.e., if we have $z_{k} \geq z_{k - 1}$ and
$2(z_{1} + ... + z_{k - 2}) + z_{k} < 3$), then the only possibility is $k=2$
and $z_{2} \leq 2$. In this case consider a $z_{2}$-step algorithm $A^{\prime
\prime}$ of class $\cal{A}$$_{1}$. The overhead of algorithm $A^{\prime\prime}$
is $\cal{O}$$(A^{\prime \prime}) = \max\{\frac{n + z_{2} - 1}{n - 1},
\frac{3z_{2} + 3}{z_{2} + 3}, \frac{z_{2} + 2n - 3}{z_{2} + n}\}$ {when $z_{2} \leq n-3$, and $\cal{O}$$(A^{\prime \prime}) = \frac{n+z_{2}-1}{n-1}$
when $z_{2} = n-2$}, by Lemma 2.2.

We have $\frac{{\cal{C}}(A, F)}{opt(F)} = \frac{4z_1+3z_2+ 3}{2z_1 + z_{2} + 3} \geq
\frac{3z_{2} + 3}{z_{2} + 3}$ {$\geq \frac{n+z_{2}-1}{n-1}$},
because $z_2 \leq 2$.

Let $F_1$ be the fault
configuration consisting of one fault for which $y = z_{2} + 1$ and $x = n - z_{2} - 2$.
We have ${\cal{C}}(A, F_1) =2z_1+3z_2+2x+1$ and $opt(F_1)=\min\{2z_2+x+2, 2x+z_2+1\}$. 
Hence ${\cal{C}}(A, F_1) \geq z_2+2n-3$ and $opt(F_1) \leq z_2+n$, and thus 
$\frac{{\cal{C}}(A, F_1)}{opt(F_1)}  \geq \frac{z_{2} + 2n - 3}{z_{2} + n}$.

Finally, consider the empty fault configuration $\emptyset$.
We have ${\cal{C}}(A, \emptyset) =2z_1+ z_2+n-1 \geq n+z_2-1$ and $opt(\emptyset)=n-1$.
Hence $\frac{{\cal{C}}(A, \emptyset)}{opt(\emptyset)}  \geq \frac{n
+ z_{2} - 1}{n - 1}$.

It follows that in the case when $z_{k} \geq z_{k - 1}$ and $2(z_{1} + ... +
z_{k - 2}) + z_{k} < 3$, we have $\cal{O}$$(A^{\prime \prime}) \leq
{\cal{O}}(A)$. This proves the lemma.
\end{proof}

\begin{proposition}\label{lb}
Every exploration algorithm of the ring has overhead at least:
$$
\begin{tabular}{rll}
	$\bullet$ & $\frac{2n - 3}{n}$ 		& for\ \ $4 \leq n \leq 5$ \\
	\\
	$\bullet$ & $\frac{3}{2}$  			& for\ \ $6 \leq n \leq 7$ \\
	\\
	$\bullet$ & $\frac{2n - 2}{n + 1}$  & for\ \ $8 \leq n \leq 19$ \\	
	\\
	$\bullet$ & $\frac{9}{5}$  			& for\ \ $20 \leq n \leq 23$ \\	
	\\
	$\bullet$ & $\frac{2n - 1}{n + 2}$  & for\ \ $n \geq 24$ \\		
\end{tabular} 
$$
\end{proposition}

\begin{proof}
By Lemma \ref{old} we can restrict attention to algorithms of class $\cal{A}$$_{0}$
and $i$-step algorithms of class $\cal{A}$$_{1}$, for  $i \leq 2$.

Let $4 \leq n \leq 5$. By Lemma \ref{class 0}, every algorithm of class $\cal{A}$$_{0}$ has overhead
at least $\frac{2n - 3}{n}$. By Corollary \ref{cor}, every $1$-step algorithm of class $\cal{A}$$_{1}$
has overhead at least $\frac{3}{2}$ and
every $2$-step algorithm of class $\cal{A}$$_{1}$
has overhead at least 
 $\frac{5}{3}$.
Since $\frac{2n - 3}{n} \leq \frac{3}{2} < \frac{5}{3}$, the overhead of any exploration algorithm is at least
$\frac{2n - 3}{n}$ in this case.

Let $6 \leq n \leq 7$. By Lemma \ref{class 0}, every algorithm of class $\cal{A}$$_{0}$ has overhead
at least $\frac{2n - 3}{n}$. By Corollary \ref{cor}, every $1$-step algorithm of class $\cal{A}$$_{1}$
has overhead at least $\frac{3}{2}$ and
every $2$-step algorithm of class $\cal{A}$$_{1}$
has overhead at least $\frac{9}{5}$. Since $\frac{3}{2} \leq \frac{2n -
3}{n} < \frac{9}{5}$, the overhead of any exploration algorithm is at least $\frac{3}{2}$ in this case.

Let $8 \leq n \leq 19$. By Lemma \ref{class 0}, every algorithm of class $\cal{A}$$_{0}$ has overhead
at least  $\frac{2n - 3}{n}$.
By Corollary \ref{cor}, every $1$-step algorithm of class $\cal{A}$$_{1}$
has overhead at least
$\frac{2n - 2}{n + 1}$ and
every $2$-step algorithm of class $\cal{A}$$_{1}$
has overhead at least
 $\frac{9}{5}$.
Since $\frac{2n - 2}{n + 1} < \frac{2n - 3}{n} \leq \frac{9}{5}$ for $8 \leq n
\leq 15$ and $\frac{2n - 2}{n + 1} \leq \frac{9}{5} < \frac{2n - 3}{n}$ for $16
\leq n \leq 19$, the overhead of any exploration algorithm is at least $\frac{2n - 2}{n + 1}$ in this case.

Let $20 \leq n \leq 23$. By Lemma \ref{class 0}, every algorithm of class $\cal{A}$$_{0}$ has overhead
at least  $\frac{2n - 3}{n}$. By Corollary \ref{cor}, every $1$-step algorithm of class $\cal{A}$$_{1}$
has overhead at least $\frac{2n - 2}{n + 1}$ and
every $2$-step algorithm of class $\cal{A}$$_{1}$
has overhead at least $\frac{9}{5}$.
Since $\frac{9}{5} < \frac{2n - 2}{n + 1} < \frac{2n - 3}{n}$, 
the overhead of any exploration algorithm is at least $\frac{9}{5}$ in this case.

Let $n \geq 24$. By Lemma \ref{class 0}, every algorithm of class $\cal{A}$$_{0}$ has overhead
at least  $\frac{2n - 3}{n}$. 
By Corollary \ref{cor}, every $1$-step algorithm of class $\cal{A}$$_{1}$
has overhead at least $\frac{2n - 2}{n + 1}$
and
every $2$-step algorithm of class $\cal{A}$$_{1}$
has overhead at least $\frac{2n - 1}{n + 2}$.  Since $\frac{2n - 1}{n +
2} < \frac{2n - 2}{n + 1} < \frac{2n - 3}{n}$, 
the overhead of any exploration algorithm is at least  $\frac{2n - 1}{n + 2}$ in this case.
\end{proof}

Propositions \ref{ring} and  \ref{lb} imply:

\begin{theorem}\label{perfect}
Algorithm {\tt Ring} is perfectly competitive for an arbitrary ring.
\end{theorem}

{\bf Remark.}
It is interesting to compare our scenario of exploration of an $n$-node ring with possibly faulty links to the scenario from \cite{DePe} of exploration of a line of known size $n$ (without faults) by an agent starting in an unknown node of the line. (This scenario was called {\em exploration with an unanchored map} in \cite{DePe}). Notice that the latter scenario is equivalent to our ring scenario with exactly one fault. In  \cite{DePe}, an algorithm with optimal overhead was given and proved to have overhead $\sqrt{3}$. Since there were no faults in \cite{DePe}, the overhead involved a maximum ratio over all positions of the starting node.
This algorithm, if applied in our case of rings with possibly many faulty links, would have a huge overhead: arbitrarily close to 3, for large $n$, and thus much worse than DFS. This is due to the fact that, in the scenario from \cite{DePe}, the adversary can only place the starting node arbitrarily, but has no  power of changing the size of the line, as this is known to the agent. By contrast, in a ring with arbitrary fault configurations, the adversary can place arbitrarily the two faults that define the size of the connected component containing the starting node. This is what makes the optimal algorithm from \cite{DePe} so bad in our present scenario.

\section{DFS exploration of hamiltonian graphs}

Since the agent has a faithful map of the graph, for most graphs there are many DFS (Depth-First-Search) 
exploration algorithms of the fault-free component of the graph, depending on the order in which free ports at each visited node are explored.
Thus it is appropriate to speak of algorithms of the class $\cal{DFS}$, each of whose members is specified by giving a permutation of ports
at each node that determines this order. Hence we adopt the following definition.

Fix an arbitrary graph $G=(V,E)$ and any fault configuration $F$ in $G$. Let $v$ be the starting node of the agent. For any node $w$ of $G$, let $\sigma _w$ be a permutation of ports at $w$. Let $\alpha=( \sigma _w : w \in V)$. $DFS (\alpha)$ is the algorithm specified by calling the following recursive procedure 
at node $v$.

\begin{center}
\fbox{
\begin{minipage}{13cm}

{\bf Procedure} {\tt Explore}$(w)$

Let $(i_1,\dots ,i_k)$ be the sequence of free ports at $w$ listed in order given by $\sigma _w$.\\
Let  $(v_1, \dots, v_k)$ be the list of neighbors of $w$ corresponding to these ports.\\
For any $r=1,\dots ,k$, let $j_r$ be the port at $v_r$ corresponding to the edge $\{v_r,w\}$.\\
mark$(w)$\\
{\bf for} $r=1$ {\bf to} $k$ {\bf do}\\
\hspace*{1cm}{\bf if} $v_r$ is unmarked {\bf then}\\ 
\hspace*{1cm}\hspace*{1cm}take port $i_r$\\
\hspace*{1cm}\hspace*{1cm}{\tt Explore}$(v_r)$\\
\hspace*{1cm}\hspace*{1cm}{\bf if} there are still unmarked nodes {\bf then} take port $j_r$

\end{minipage}
}
\end{center}

The class $\cal{DFS}$ is defined as the class of algorithms $DFS (\alpha)$, for all sequences of permutations $\alpha=( \sigma _w : w \in V)$.
Note that for any algorithm of the class $\cal{DFS}$, the route of the agent is a tour of some spanning tree of the fault-free component $C$ of the graph,
containing the starting node and corresponding to the fault configuration $F$, in which the agent stops as soon as the last node of the component is visited. Hence the edge of the tree
corresponding to the  leaf that is visited last is traversed only once. Consequently, for any $\alpha$ and any $F$,  we have ${\cal C}(DFS(\alpha),F) \leq 2m-3$,
where $m$ is the size of the component $C$. The total computation time of any algorithm of the class $\cal{DFS}$ is linear in the number of all ports, i.e., also linear in the number of edges of the
explored graph. 

Our aim is to show that if the graph $G$ is hamiltonian then any algorithm of the class $\cal{DFS}$ has overhead not much larger than that of a perfectly competitive algorithm. Note that a hamiltonian graph must have at least three nodes, and hence we assume in the sequel that the size $n$ of the graph is at least 3. 
We first prove the following lemma.

\begin{lemma}\label{dfs}
Let $G$ be any hamiltonian graph of size $n \geq 3$. For any algorithm $A$ of the class $\cal{DFS}$ and any starting node, ${\cal O}_{A} \leq \frac{2n-4}{n-1}$.
\end{lemma} 

\begin{proof}
Consider any fault configuration $F$ and any starting node $v$. Let $C$ be the fault-free component, corresponding to the configuration $F$ and to the node $v$.
Let $m$ be the size of $C$. Consider two cases.

{\bf Case 1.} There exists a hamiltonian path in $C$ with $v$ as an endpoint.

Let $(w_1,\dots ,w_m)$, where $w_1=v$, be a hamiltonian path in $C$. $C$
contains all edges of this path and possibly some additional edges
$\{w_i,w_j\}$. For $m=2$, we have ${\cal
C}(A,F) = opt(F) = 1$. Assume $m\geq 3$; we have 
$opt(F)=m-1$. Let $T$ be the spanning tree of
$C$, corresponding to the route of the agent executing algorithm $A$. $T$ has
$m-1$ edges. As previously remarked, the last-traversed edge of $T$ is traversed
only once. We show that the first-traversed edge is also traversed only once.
Let $\{w_1,w_j\}$ be this edge. Suppose that this edge is traversed by the agent
a second time. This traversal must be from $w_j$ to $w_1$. This implies that at
the time $t$ of the visit of $w_j$ immediately preceding this traversal, there
are still non-visited nodes. Let $w_s$ be the node non-visited at time $t$, with
the largest index $s<j$ or with the smallest index $s>j$. If $s<j$ then
$w_{s+1}$ was visited before time $t$ and if $s>j$ then $w_{s-1}$ was visited
before time $t$. In the first case the port at $w_{s+1}$ corresponding to the
edge $\{w_{s+1},w_s\}$ is free, and in the second case the port at $w_{s-1}$
corresponding to the edge $\{w_{s-1},w_s\}$ is free.
This contradicts the fact that, according to algorithm $A$, the agent should
visit $w_s$ before time $t$.

Hence both the first-traversed edge and the last-traversed edge are traversed
only once. This implies that the cost of algorithm $A$ is at most $2m-4$,
{for $m \geq 3$}, in this case.

{\bf Case 2.} There is no hamiltonian path in $C$ with $v$ as an endpoint.

In this case $opt(F)\geq m$ and the cost of algorithm $A$ is at most $2m-3$ because the last-traversed edge is traversed only once.

Since $\frac{2m-4}{m-1}\geq \frac{2m-3}{m}$ for $m\geq 3$,
it follows that in both cases $\frac{{\cal C}(A,F)}{opt(F)} \leq \frac{2m-4}{m-1}$, whenever the component $C$ is of size $m$. Since $\frac{2m-4}{m-1}$ is an increasing function of $m$, this value is largest for $m=n$, which concludes the proof.
\end{proof} 

The following is the main result of this section.

\begin{theorem}
For any hamiltonian graph $G=(V,E)$, any starting node $v$, and any algorithm $A$ of the class ${\cal DFS}$, the ratio between the overhead
of $A$ and the overhead of a perfectly competitive algorithm on $G$ starting at $v$ is at most 10/9. Moreover, for $n \geq 24$, this ratio is less than 1.06.
\end{theorem}

\begin{proof}
Let $B$ be a perfectly competitive algorithm on $G$, starting at $v$.
Let $R$ be any hamiltonian cycle in $G$. By definition, $R$ contains the starting node $v$. Let $F^*$ be the set of all edges in $G$ except those in $R$.
Let $x=\mbox{max}_{F^* \subseteq F \subseteq E} \frac{{\cal C}(B,G,v,F)}{opt(G,v,F)}$. By definition, ${\cal O}_{B,G,v} \geq x$. On the other hand, the execution of any
algorithm for a fault configuration $F$ containing $F^*$ can be considered as an execution of an algorithm on the ring $R$ for the fault configuration $F \setminus F^*$.
Since the family of fault configurations in $G$ containing $F^*$ is equal to the family of fault configurations $F'\cup F^*$, where $F'$ is a fault configuration in $R$, 
we get $x\geq {\cal O}_{{\tt Ring},R,v}$, for otherwise Algorithm {\tt Ring} would not be perfectly competitive on $R$, contradicting Theorem \ref{perfect}. 
It follows that ${\cal O}_{B,G,v} \geq {\cal O}_{{\tt Ring},R,v}$. By Proposition \ref{lb}, ${\cal O}_{B,G,v}$ is at least:

$$
\begin{tabular}{rll}
	$\bullet$ & $\frac{2n - 3}{n}$ 		& for\ \ $4 \leq n \leq 5$ \\
	\\
	$\bullet$ & $\frac{3}{2}$  			& for\ \ $6 \leq n \leq 7$ \\
	\\
	$\bullet$ & $\frac{2n - 2}{n + 1}$  & for\ \ $8 \leq n \leq 19$ \\	
	\\
	$\bullet$ & $\frac{9}{5}$  			& for\ \ $20 \leq n \leq 23$ \\	
	\\
	$\bullet$ & $\frac{2n - 1}{n + 2}$  & for\ \ $n \geq 24$ \\		
\end{tabular} 
$$

By Lemma \ref{dfs}, the ratio $\frac{{\cal O}_{A,G,v}} {{\cal O}_{B,G,v}}$ is at most: 

\
$$
\begin{tabular}{rll}
	$\bullet$ & $1$ & when\ \ $n= 3$ \\
	\\
	$\bullet$ & $\frac{\frac{2n-4}{n-1}}{\frac{2n-3}{n}} \leq \frac{15}{14} 
	\approx 1.0714$ & when\ \ $4 \leq n \leq 5$ \\
	\\
	$\bullet$ & $\frac{\frac{2n-4}{n-1}}{\frac{3}{2}} \leq \frac{10}{9} 
	\approx 1.1111$ & when\ \ $6 \leq n \leq 7$ \\
	\\
	$\bullet$ & $\frac{\frac{2n-4}{n-1}}{\frac{2n-2}{n+1}} \leq \frac{54}{49} 
	\approx 1.1020$ & when\ \ $8 \leq n \leq 19$ \\
	\\
	$\bullet$ & $\frac{\frac{2n-4}{n-1}}{\frac{9}{5}} \leq \frac{35}{33} 
	\approx 1.0606$ & when\ \ $20 \leq n \leq 23$ \\
	\\
	$\bullet$ & $\frac{\frac{2n-4}{n-1}}{\frac{2n-1}{n+2}} \leq \frac{1144}{1081} 
	\approx 1.0583$ & when\ \ $n \geq 24$ \\
\end{tabular} 
$$

\

Since $1 < \frac{1144}{1081} < { \frac{35}{33}} < \frac{15}{14} <
{ \frac{54}{49}} < \frac{10}{9}$, the ratio between the overhead
of $A$ and the overhead of a perfectly competitive algorithm on $G$ starting at $v$ is at most 10/9. 
Moreover, for $n \geq 24$, this ratio is less than 1.06.
\end{proof}

It is natural to ask if the bound $\frac{2n-4}{n-1}$ on the overhead of any algorithm of the class ${\cal DFS}$, obtained in Lemma \ref{dfs}, is tight.
This is not the case for all hamiltonian graphs, as witnessed by the example of the ring. By Lemma \ref{class 0}, the overhead of any algorithm of the class ${\cal DFS}$ on the ring
is $\frac{2n-3}{n}<\frac{2n-4}{n-1}$. Hence a natural question is whether this bound is tight for some hamiltonian graphs. The next proposition shows that the answer to this question is positive.

\begin{proposition}
The overhead of any algorithm $A$ of the class ${\cal DFS}$ is exactly
$\frac{2n-4}{n-1}$ for the complete graph of size $n\geq 3$.
\end{proposition}

\begin{proof}
Let $A$ be any algorithm of the class ${\cal DFS}$ starting at a node $v$ of the complete graph $K_n$.
In view of Lemma \ref{dfs}, it is enough to show a fault configuration $F$, for which $ \frac{{\cal C}(A,K_n,v,F)}{opt(K_n,v,F)} = \frac{2n-4}{n-1}$.
This configuration depends on the set of permutations $\alpha=( \sigma _w : w \in V)$ for which $A=DFS(\alpha)$.
Let $\sigma _v=(i_1,\dots ,i_{n-1})$ and let $w$ be the node corresponding to port $i_1$. Let $\sigma _w=(j_1,\dots ,j_{n-1})$ and let $j_k$ be the port number at $w$
corresponding to the edge $\{w,v\}$. Let $a=\min{\{1,\dots, n-1\} \setminus \{k\}}$ and let $b=\min{\{1,\dots, n-1\} \setminus \{k,a\}}$.
Let $u'$ be the node corresponding to port $j_a$ at node $w$ and let $u''$ be the node corresponding to port $j_b$ at node $w$. We define
$v_1=v$, $v_{n-2}=u'$, $v_{n-1}=w$, $v_n=u''$, and all nodes other than $v,u',w,u''$ are called arbitrarily $v_2,\dots, v_{n-3}$. 
Let $P$ be the set of edges $\{\{v_1,v_{n-1}\},\{v_1,v_2\}, \{v_2,v_3\},\dots, \{v_{n-1},v_n\}\}$ (see Fig. 1).

\begin{figure}[h]
	\centering
	\begin{tikzpicture}[scale=1.3] 
		
		\draw[]
			(0,0) -- (4,0)
			(6,0) -- (10,0);	
		
		\draw[dotted]
			(4,0) -- (6,0);	
		
		\draw[rounded corners=10pt] 
			(0,0) -- (0.25, 1) -- (2, 1) -- (7.75, 1) -- (8, 0);			
		
		\node [scale=1, black] at (0,0) {\textbullet};
		\node [scale=1, black] at (2,0) {\textbullet};
		\node [scale=1, black] at (4,0) {\textbullet};
		\node [scale=1, black] at (6,0) {\textbullet};
		\node [scale=1, black] at (8,0) {\textbullet};
		\node [scale=1, black] at (10,0) {\textbullet};

		\node [scale=1, black, anchor=north] at (0,0) {$v = v_{1}$};
		\node [scale=1, black, anchor=north] at (2,0) {$v_{2}$};
		\node [scale=1, black, anchor=north] at (4,0) {$v_{3}$};
		\node [scale=1, black, anchor=north] at (6,0) {$u' = v_{n-2}$};
		\node [scale=1, black, anchor=north] at (8,0) {$w = v_{n-1}$};
		\node [scale=1, black, anchor=north] at (10,0) {$u'' = v_{n}$};
		
		\node [scale=1, red] at (-0.1,0.4) {$i_{1}$};
		\node [scale=1, red] at (7.6,0.2) {$j_{a}$};
		\node [scale=1, red] at (8.4,0.2) {$j_{b}$};
		\node [scale=1, red] at (8.1,0.55) {$j_{k}$};

	\end{tikzpicture}
	\caption{The set of edges $P$}
    \label{fig:EdgesOfP}	
\end{figure}
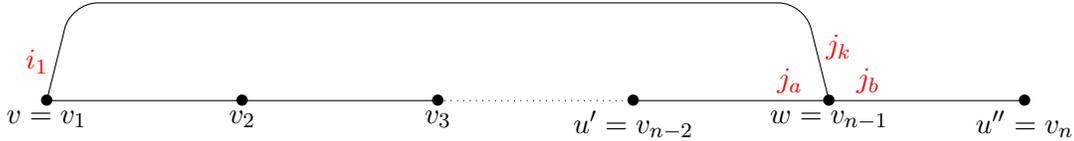

Consider the fault configuration $F$ consisting of all edges of $K_n$
except those in $P$. For this fault configuration the optimal cost is $n-1$, corresponding to traveling the path $v_1, v_2,\dots ,v_n$. The agent executing algorithm $A$ first takes port
$i_1$ going to node $w$, then takes port $j_a$ which is the first free port at $w$ according to permutation $\sigma _w$, apart from the port $j_k$ by which the agent has just come to $w$. At this point the agent
is at $u'=v_{n-2}$. Next the agent travels the path $(v_{n-2},v_{n-3},\dots, v_2)$, backtracks by the reverse path $(v_2,v_3,\dots,v_{n-2})$ to $u'=v_{n-2}$, backtracks to 
$w=v_{n-1}$, and finally goes to $u''=v_n$. The total cost is $2n-4$. Hence $ \frac{{\cal C}(A,K_n,v,F)}{opt(K_n,v,F)} = \frac{2n-4}{n-1}$.
\end{proof}

\section{Conclusion}

We designed a perfectly competitive exploration algorithm for rings,  and we proved that, for all hamiltonian graphs, the overhead of {\em every} algorithm of the class ${\cal DFS}$ exceeds that of a perfectly competitive algorithm for the given graph by less than
a factor of 10/9. The latter result does not hold for arbitrary graphs. This is best witnessed by the case of a line, for which the starting node is at distance 1 from one of the endpoints. In this case 
the perfectly competitive algorithm has overhead 1 (as shown in \cite{MaPe}), and the ``wrong'' DFS (i.e., the one that goes first to the farther endpoint) has overhead $\frac{2n-1}{n+1}$, for lines of length $n$, i.e., it is arbitrarily close to 2 for long lines.
This still leaves the hope that the ``best'' member of the class ${\cal DFS}$, i.e., the DFS algorithm with the smallest overhead, for any given graph, may have overhead similar to that of a perfectly
competitive algorithm for this graph. This might be true, but finding this best DFS (a trivial task in the case of the line) may be very challenging for arbitrary graphs. We know neither
if the ``best'' DFS for an arbitrary graph has an overhead similar to that of a a perfectly
competitive algorithm for this graph, nor, even if this is the case, how to find this particular DFS, or any other of similar overhead. 

Going beyond the class ${\cal DFS}$, the most important question seems to be, whether there exists a perfectly competitive algorithm for arbitrary graphs,
whose total computation time is polynomial in the size of the graph. Unless such an algorithm is constructed, 
an important challenge would be finding, for an arbitrary graph, some exploration algorithm whose overhead exceeds that of a perfectly competitive algorithm
for the given graph by a small factor. In view of the present paper and of \cite{MaPe} the situation is the following. For the class  of hamiltonian graphs and for the class of trees, such algorithms exist, and the ratio
is at most 10/9 in the first case and at most 9/8 in the second case. For these two classes of graphs, the algorithms performing so well are very simple, and the total computation time
used during exploration is linear in the number of edges of the explored graph.
(Since {\em any} algorithm of the class ${\cal DFS}$ on hamiltonian graphs is good for our purpose, we can take, e.g.,
the one induced by the increasing permutation of ports at each node.) Is it possible to design, for an arbitrary graph, an exploration algorithm whose overhead exceeds that of a perfectly competitive algorithm for this graph by a factor of, say, at most 1.15, and for which the total computation time
is polynomial in the size of the graph?

Other interesting variants of the problem of exploration in the presence of faulty edges are when the agent knows the exact value of the number of faults, or some upper bound on this number.

\end{document}